\documentclass[12pt]{article}
\usepackage[utf8]{inputenc}
\usepackage[english]{babel}
\usepackage{hyperref}
\usepackage{textcomp}
\usepackage[autostyle]{csquotes}

\usepackage[backend=biber, maxbibnames=99]{biblatex}
\usepackage{comment}
\usepackage{authblk}
\usepackage{amsthm}
\usepackage{amsfonts}
\usepackage{mathrsfs}
\usepackage{amssymb}
\usepackage{graphicx}
\usepackage{geometry}
\usepackage{subfiles}
\usepackage{subfig}
\usepackage{mathtools}
\usepackage{pgfplots}
\pgfplotsset{/pgf/number format/use comma,compat=newest}
\usepackage{url}
\usepackage{tikz}
\usepackage{bbm}
\usepackage{floatflt}
\usetikzlibrary{arrows,decorations.pathmorphing,backgrounds,positioning,fit,petri}

\newcommand{\newinf}{\mathop{\mathrm{inf}\vphantom{\mathrm{sup}}}}
\newcommand{\E}{\mathbb{E}}
\newcommand{\R}{\mathbb{R}}
\newcommand{\N}{\mathbb{N}}

\renewcommand{\d}{\mathrm{d}}

\newcommand{\virg}[1]{``#1''}

\newcommand\iid{\mathrel{\stackrel{\makebox[0pt]{\mbox{\normalfont\tiny iid}}}{\sim}}}

\newcommand\aseq{\mathrel{\stackrel{\makebox[0pt]{\mbox{\normalfont\tiny a.s.}}}{=}}}

\theoremstyle{plain} 
\newtheorem{theorem}{Theorem}

\newtheorem{lemma}[theorem]{Lemma} 
\newtheorem{proposition}[theorem]{Proposition} 

\theoremstyle{definition} 
\newtheorem{definition}{Definition}

\theoremstyle{remark} 
\newtheorem{remark}{Remark}

\title{The solution of the deep Boltzmann machine on the Nishimori line}
\author[1]{Diego Alberici}
\author[2]{Francesco Camilli}
\author[2]{Pierluigi Contucci}
\author[2]{Emanuele Mingione}
\affil[1]{Communication Theory Laboratory, EPFL, Switzerland}
\affil[2]{Dipartimento di Matematica, Università di Bologna, Italy}

\newcommand{\be}{\begin{equation}}
\newcommand{\ee}{\end{equation}}

\begin{document}

\maketitle

\begin{abstract}
    The deep Boltzmann machine on the Nishimori line with a finite number of layers is exactly solved by a theorem that expresses its pressure through a finite dimensional variational problem of \emph{min-max} type. In the absence of magnetic fields the order parameter is shown to exhibit a phase transition whose dependence on the geometry of the system is investigated.
\end{abstract}

\textbf{keywords}: Deep Boltzmann machines, Nishimori line, Multi-species spin-glass, Replica symmetry\\

\section{Introduction}
A deep (restricted) Boltzmann machine can be considered as a special case of the mean field multi-species spin glass model introduced in \cite{MSK_original}, further studied in \cite{bates,panchenko_multi-SK}. Specifically the set of spins is arranged into a geometry made of consecutive layers and only interactions among spins belonging to adjacent layers are allowed. In particular intra-layer interactions are forbidden. Such architectural assumption makes it impossible to fulfill the positivity hypothesis under which the results of \cite{MSK_original,panchenko_multi-SK} were obtained. In fact the positivity property, encoded in an elliptic condition, requires dominant intra-group interaction with respect to inter-group ones. 
While the general deep (restricted) Boltzmann machine is still an unsolved problem (see nevertheless \cite{Alberici2020AnnealingAR,noi_deep2,AuChe,baik2020,bgg,genovese,mourrat20,mourrat200} for centered Gaussian interactions), we present here its exact and rigorous solution in a subregion of the phase space known as Nishimori line. In a previous paper \cite{us} we have fully solved the elliptic multi-species model on the Nishimori line, where the property of replica symmetry, \textit{i.e.} the concentration of the overlap, was shown to hold. Such property indeed is fully general on the Nishimori line, see \cite{jean-dmitry} on this respect, and does not rely on any positivity assumption of the interactions. 
While the positivity properties carry with them the typical bounds of Guerra's method \cite{Guerra_upper_bound,interp_guerra_2002}, here the technical support to control and solve the model is based on the presence, on the Nishimori line, of a set of identities relating magnetizations and overlaps expectations \cite{Nishi_id_PC,Hide_original,nishimori01} and correlation inequalities \cite{contucci_morita_nishimori}. Our work provides the first exact solution of a disordered Statistical Mechanics model in a deep architecture and describes how the relative size of the layers influences the phase transition.

The relevance of the Nishimori line is twofold. On one side it provides the possibility to investigate the replica symmetric phase of the model through an exact solution for arbitrary strength of the interactions. On the other side it represents a bridge between a class of inference problems and Statistical Physics \cite{nishimori01}. For instance the Sherrington-Kirkpatrick model on the Nishimori line corresponds to the Wigner Spiked model in the inference Bayesian optimal setting with binary signals \cite{adaptive,wigner-wishart}. Analogously, any multi-species mean-field model on the Nishimori line can be seen as a spatially coupled Wigner spiked model first introduced and studied in \cite{Lenka,Jean-lenka2}. From the inference point of view here we deal with a deep spatially coupled Wigner spiked model with $K$ layers, which in the case $K=2$ coincides with the Wishart model (rank-one non-symmetric matrix estimation \cite{wigner-wishart}).

The paper is organized as follows. In Section \ref{def&results} we introduce the model and we present the main results in three theorems. Section \ref{preliminary_sect} is a collection of tools and preliminary results, starting form the Nishimori identities and the correlation inequalities, up to the adaptive interpolation method. The proofs are contained in Section \ref{proofs_sect} and Section \ref{sect_five} collects some conclusions and perspectives.

\section{Definitions and results}\label{def&results} 
Consider a set of sites with cardinality $N$, divide it into $K$ disjoint subsets, called \emph{layers} and denoted by $\{L_r\}_{r=1,\dots,K}$ with cardinality $|L_r|=N_r$ and $\sum_{r=1}^K N_r=N$. To each site $i$ we associate an Ising spin $\sigma_i$ and we denote $\sigma=(\sigma_1,\dots,\sigma_N)$ a configuration of spins belonging to the space $\Sigma_N=\{+1,-1\}^N$. The Hamiltonian of the model is defined as:
\begin{align}
    \label{H_DBM}
    &H_N(\sigma):=-\sum_{r,s=1}^K\sum_{(i,j)\in L_r\times L_s}\tilde{J}_{ij}^{rs}\sigma_i\sigma_j
    -\sum_{r=1}^K\sum_{i\in L_r}\tilde{h}^r_i\sigma_i
\end{align}
where the interaction coefficients and the external fields are independent Gaussian random variables distributed as follows
\begin{align}
    \label{gaussian_couplings_NL}
    &\tilde{J}_{ij}^{rs}\iid
    \mathcal{N}\left(\frac{\mu_{rs}}{2N},\frac{\mu_{rs}}{2N}\right),\;\quad \tilde{h}^r_i\iid
    \mathcal{N}(h_r,h_r)\,,
\end{align}
and the matrix $\mu:=(\mu_{rs})_{r,s=1,\dots,K}$ and the vector $\mathbf{h}:=(h_r)_{r=1,\dots,K}$ have non-negative entries. Furthermore $\mu$ has the following tridiagonal structure:
\begin{align}\label{structure_interactions_DBM}
    \mu=\begin{pmatrix}
    0 & \mu_{12} & 0 &  \cdots& 0 \\
    \mu_{21}& 0 & \mu_{23}&\cdots &0 \\
    0 & \mu_{32}& 0 & \ddots& 0 \\
     \vdots& \vdots & \ddots & \ddots& \mu_{K-1,K}\\
    0& 0&0 &\mu_{K,K-1} & 0 
    \end{pmatrix}\,
\end{align}
and is assumed to be symmetric without loss of generality. 
The geometrical architecture of the model is illustrated in \figurename\ref{mainfig}.

\begin{figure}[h]
    \centering
    \includegraphics[width=.9\textwidth]{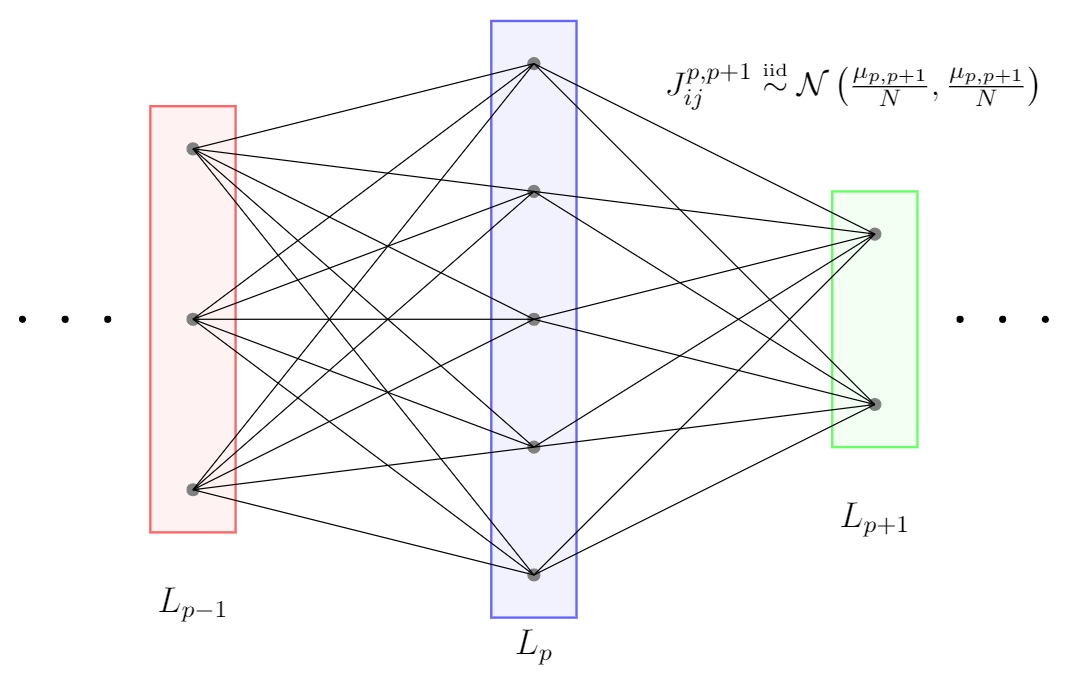}
    \caption{Graph of the interactions between layers.}
    \label{mainfig}
\end{figure}
We point out that the very special choice of the Gaussian distribution \eqref{gaussian_couplings_NL}, having mean values and variances tied to be the same, is called \textit{Nishimori line} in Physics literature \cite{nishimori01}. We will recall correlation identities and inequalities holding on the Nishimori line in the next Section.


We denote
\begin{align}
    &m_r(\sigma):=\frac{1}{N_r}\sum_{i\in L_r}\sigma_i\,,\quad\qquad q_r(\sigma,\tau):=\frac{1}{N_r}\sum_{i\in L_r}\sigma_i\tau_i \;;\\[4pt]
    &\mathbf{m}(\sigma):=(m_r(\sigma))_{r=1,\dots,K}\,,\quad
    \mathbf{q}(\sigma,\tau):=(q_r(\sigma,\tau))_{r=1,\dots,K}
\end{align}
with bold characters here and below standing for vectors and $\sigma,\tau\in\Sigma_N=\{-1,1\}^N$.
We also set
\begin{align}
    \Delta:=(\alpha_r\mu_{rs}\alpha_s)_{r,s=1,\dots,K}\,,\;\quad\hat{\alpha}:=\text{diag}(\alpha_1,\alpha_2,\dots,\alpha_K)\,,\;
\end{align}
where $\alpha_r=N_r/N$ are called the \emph{form factors}. $\Delta$ is the \emph{effective interaction matrix} and encodes all the information on the interactions of the system.  For later convenience we introduce also the matrix
\begin{align}\label{M}
M := (\mu_{rs}\alpha_s)_{r,s=1,\dots,K}
\end{align}
We notice that $\Delta$ and $M$ are tridiagonal matrices too.

It is useful to express the Hamiltonian \eqref{H_DBM} in terms of centered Gaussian random variables plus a deterministic term (in vector notation):
\begin{align}
    \label{H_MSK_NL-bis}
    H_N(\sigma)=-\frac{1}{\sqrt{2N}}\sum_{r,s=1}^K\sum_{(i,j)\in L_r\times L_s}J_{ij}^{rs}\sigma_i\sigma_j-\sum_{r=1}^K\sum_{i\in L_r}h^r_i\sigma_i-\frac{N}{2}(\mathbf{m},\Delta\mathbf{m})-N(\hat{\alpha}\mathbf{h},\mathbf{m})
\end{align}
where $\left(\cdot,\cdot\right)$ denotes the Euclidean inner product in $\mathbb{R}^K$ and
\begin{align}
    J_{ij}^{rs}\iid\mathcal{N}\left(0,\mu_{rs}\right),\;\quad
    h_i^r\iid\mathcal{N}(0,h_r)\;.
\end{align}
The random term in \eqref{H_MSK_NL-bis} corresponds to the Hamiltonian studied in \cite{noi_deep2}, but the addition of a deterministic part changes the properties of the model.

We denote the random pressure per particle by
\begin{align}
    \label{press}
    &p_N:=\frac{1}{N}\log\sum_{\sigma\in\Sigma_N}\exp\left(-H_N(\sigma)
    \right)
\end{align}
and its quenched average by
\begin{align}
    \label{Q_press}
    &\bar{p}_N(\mu,\mathbf{h}):=\mathbb{E}\,p_N \,, 
\end{align}
where $\mathbb E$ is the expectation with respect to all the Gaussian random variables.
 \begin{remark}\label{finitevolume}
While throughout this paper we keep the form factors $\alpha_r$'s constant as $N\to\infty$, all the results hold also under the weaker hypothesis that $N_r/N\to\alpha_r\in (0,1)$ (see also Remark \ref{remark_alpha->0}, Sect. 4.3 for vanishing $\alpha_r$'s).
Indeed any vanishing correction to $\alpha_r$ doesn't change the thermodynamic limit of the quenched pressure density \eqref{Q_press}. This can be seen proving by interpolation method that at given $N$ the quenched pressure is a Lipschitz function of $\Delta$ w.r.t. the entrywise matrix norm $\sum_{r,s\leq K}|\Delta_{r,s}|$.
\end{remark}

The (random) Boltzmann-Gibbs average will be denoted by
\begin{align}\label{gibbs_average}
    \langle\cdot\rangle_N:=\frac{\sum_{\sigma\in\Sigma_N}e^{-H_N(\sigma)}(\cdot)}{Z_N},\,\quad Z_N:=\sum_{\sigma\in\Sigma_N}e^{-H_N(\sigma)}\,.
\end{align}
To help the presentation we will occasionally make explicit the dependence of the Boltzmann-Gibbs measure on further parameters by using sub and superscripts, for instance $\langle\cdot\rangle_{N,t}^{(\epsilon)}$.
In the previous definitions \eqref{press}-\eqref{gibbs_average} we have chosen to reabsorb the inverse absolute temperature $\beta$ in the parameters $\mu_{rs}$ and $h_r$.
The first result of this paper is the computation of the random pressure \eqref{press} in the thermodynamic limit. 

\begin{theorem}[Solution of the model]
\label{main_theorem} The random pressure \eqref{press} of a $K$-layer deep Boltzmann machine on the Nishimori line converges almost surely in the thermodynamic limit and its value is given by a $K$-dimensional variational principle:
\begin{align}\label{var_principle_DBM}
    \lim_{N\to\infty}p_N\aseq\lim_{N\to\infty}\bar{p}_N(\mu,\mathbf{h})=\sup_{\mathbf{x}_o}\newinf_{\mathbf{x}_e}p_{var}(\mathbf{x};\mu,\mathbf{h}) \;,
\end{align}
where $\mathbf{x}_o$ and $\mathbf{x}_e$ denote the vectors of the odd and even components of the order parameter $\mathbf{x}\in[0,1)^K$ respectively,
\begin{multline}\label{p_var_DBM}
    p_{var}(\mathbf{x};\mu,\mathbf{h}):=\sum_{r=1}^K\,\alpha_r\,\psi\left((M\mathbf{x})_r+h_r
    \right)\,+\,\sum_{r=1}^K\,\frac{\Delta_{r,r+1}}{2}\,[(1-x_{r})(1-x_{r+1})-2x_rx_{r+1}]
\end{multline}and for any $x\geq 0$
\begin{align}\label{psi_def}
    \psi(x):=\mathbb{E}_z\log2\cosh\left(
    z\sqrt{x}+x\right)\;,\quad z\sim\mathcal N(0,1) \;.
\end{align}
Moreover, defining $\bar{\mathbf x}$ as the solution of the variational problem \eqref{var_principle_DBM}, we have
\be\label{exchange_lim}
\lim_{N\to\infty}\E \langle q_r \rangle_N \,=\, \lim_{N\to\infty}\E \langle m_r \rangle_N \;=\, \bar{x}_r 
\ee
for every $r=1,\dots,K\,$ and for all the points of the phase space $(\mu,\hat{\alpha},\mathbf{h})$ where $\bar{x}$ is $\mathbf{h}$-differentiable and the matrix $\Delta$ is invertible.
\end{theorem}

The proof of Theorem \ref{main_theorem} relies on the adaptive interpolation method \cite{adaptive} combined with a concentration result and with the Nishimori identities, that will be presented in the next section.
The main difference with the model solved in \cite{us} is that the matrix $\Delta$ is not definite, indeed its eigenvalues have alternating signs. This entails that the remainder identified by interpolation has not a definite sign and cannot be discarded a priori at the expense of an inequality. Moreover, the concentration of the overlap strongly depends on a notion of regularity of the path followed by the adaptive interpolation. Hence one has to carefully choose a path that is regular and allows also to exploit the convexities of the two sums involved in the functional \eqref{p_var_DBM}. 

Secondly, we focus on the properties of the consistency equation obtained from the optimization problem \eqref{var_principle_DBM} when the matrix $\Delta$ is invertible, that is when $K$ is even. The stability of the optimizers of \eqref{var_principle_DBM} is a more delicate problem with respect to the convex multi-species case \cite{us}, due to the min-max nature of the variational principle. In the following, given a square matrix $A$ we denote by $\rho(A)$ its spectral radius and by
$A^{(eo)}$ the submatrix of $A$ obtained by keeping only even rows and odd columns of $A$. An analogous definition is given for $A^{(oe)}, A^{(oo)}, A^{(ee)}$. Notice that, when $K$ is even, $\Delta^{(eo)}$ is an upper triangular $K/2\times K/2$ square matrix with non-zero diagonal elements and therefore it is invertible. Similar considerations hold for the sub-matrix $\Delta^{(oe)}=[\Delta^{(eo)}]^T$. We prove the following 

\begin{theorem}\label{phase_transition_thm}
Let $K$ be even and $\mathbf{h}=0$. 
If $\rho([M^{2}]^{(oo)})<1$ then $\mathbf{x}=0$ is the unique solution to the variational problem \eqref{var_principle_DBM}. Conversely, if $\rho([M^{2}]^{(oo)})>1$ then the solution of \eqref{var_principle_DBM} is a vector $\mathbf{x}=\bar{\mathbf{x}}(M)$ with strictly positive components satisfying the consistency equation:
\begin{align}
    x_r \,=\, \E_z\tanh\left(z\,\sqrt{(M\mathbf{x})_r\,}\,+\,(M\mathbf{x})_r\right) \quad\forall\,r=1,\dots,K
\end{align}
where $z$ denotes a standard Gaussian random variable.
\end{theorem}
The proof of Theorem \ref{phase_transition_thm} amounts to the computation of the Hessian matrix of an auxiliary function introduced later and in a check of its eigenvalues. The peculiar form of the consistency equations due to the structure \eqref{structure_interactions_DBM} plays a central role. Theorem \ref{phase_transition_thm} implies the existence of a phase transition in our model localized at zero magnetic field and unitary spectral radius as discussed in Remark \ref{rmk_phase_trans} below.
The following Proposition further clarifies the structure of the phase transition and how the system's geometry, encoded in the form factors $\alpha_r$'s, can influence it.

\begin{proposition} \label{prop:rhobound}
For any given interaction matrix $\mu\,$, 
we have
\be  \label{rhobound}
\sup_{\alpha_1,\dots,\alpha_K}
\rho\left([M^2]^{(oo)}\right) \,=\, \frac{1}{4}\,\max_{r} \mu_{r,r+1}^2 
\ee
where the $\sup$ on the l.h.s. is taken over the form factors $\alpha_1,\dots,\alpha_K\geq0\,$, $\sum_{r=1}^K\alpha_r=1\,$ and the $\max$ on the r.h.s. is taken over $r=1,\dots,K-1\,$.
Furthermore the $\sup$ on the l.h.s. of \eqref{rhobound} is attained if and only if one of the following conditions is verified:
\begin{itemize}
\item[a)] there exists $r^*\in\{1,\dots,K-1\}\,$ such that
\be \label{bestlambda1}
\alpha_{r^*} \,=\, \alpha_{r^*+1} \,=\, \frac{1}{2}
\quad,\quad \mu_{r^*,\,r^*+1}=\max_{r}\mu_{r,r+1} \;;
\ee
\item[b)]
there exists $r^*\in\{2,\dots,K-1\}\,$ such that
\be \label{bestlambda2}
\alpha_{r^*} \,=\, \alpha_{r^*-1}+\alpha_{r^*+1} \,=\, \frac{1}{2}
\quad,\quad \mu_{r^*-1,\,r^*}=\mu_{r^*,\,r^*+1}=\max_{r}\mu_{r,r+1} \;.
\ee
\end{itemize}
\end{proposition}
\begin{remark}\label{rmk_phase_trans}
For even $K$, Proposition \ref{prop:rhobound} together with Theorems \ref{main_theorem} and \ref{phase_transition_thm} show that if the interaction strengths $\mu_{r,r+1}<2$ for all $r=1,\dots,K-1$, then the magnetisations and the overlaps vanish as $N\to\infty$ for every choice of the form factors $(\alpha_1,\dots,\alpha_K)\in(0,1)^K$.
By Theorem \ref{phase_transition_thm} $\bar{\mathbf{x}}$ is not identically zero on the space of parameters $(\mu,\,\hat{\alpha})\,$, hence 
the limiting quenched pressure \eqref{var_principle_DBM} cannot be an analytic function.

Proposition \ref{prop:rhobound} also shows that 
as soon as $\mu_{r,r+1}>2$ for some $r=1,\dots,K-1$, then, by suitably localizing only two 
extensive layers near the maximal interaction (condition \eqref{bestlambda1}
), their magnetisations and overlaps turn out to be positive in the limit $N\to\infty$.

\end{remark}

Finally, we prove a uniqueness result that holds for arbitrary spectral radius.
\begin{theorem}\label{teor_uniq}
Let $h_r>0\ \forall\,r=1,\dots,K\,$. The consistency equation
\begin{equation} \label{eq:CE}
x_r \,=\, \E_z\tanh\left(z\,\sqrt{(M\mathbf{x})_r+h_r\,}\,+\,(M\mathbf{x})_r+h_r\right) \quad\forall\,r=1,\dots,K
\end{equation}
admits a unique solution $\mathbf{x}=\bar{\mathbf{x}}(M,\mathbf{h})\in(0,1)^K$.
\end{theorem}


\section{Preliminary results}\label{preliminary_sect}

\subsection{Nishimori identities and correlation inequalities}
The main thermodynamic properties of the model are consequences of a family of identities and inequalities for correlation functions that are due to the specific setting \eqref{gaussian_couplings_NL}. The identities were introduced in the original work by H. Nishimori \cite{Hide_original} while the inequalities were proved in \cite{contucci_morita_nishimori,morita2005}. The proof of the Nishimori identities can be found in the book \cite{contucci_giardina_2012} (Paragraph 2.6). In particular, for our purposes we will use the following:
\begin{align}
    \label{N_identity_1}
    &\mathbb{E}[\langle\sigma_i\rangle_N^{2n}]\,=\,
    \mathbb{E}[\langle\sigma_i\rangle_N^{2n-1}]\;, \quad n=1,2,3,\dots\\[2pt]
    \label{N_identity_2}
    &\mathbb{E}[\langle\sigma_i\sigma_j\rangle_N^2]\,=\,\mathbb{E}[\langle\sigma_i\sigma_j\rangle_N]
    \;.
\end{align}
From the previous relations it follows that on the Nishimori line magnetizations and overlaps moments coincide. This can be seen by
\begin{align}
    &\mathbb{E}[\langle q_s\rangle_N]\,=\,
    \sum_{i\in L_s}\frac{1}{N_s}\mathbb{E}[\langle\sigma_i\rangle_N^2]\,=\,
    \sum_{i\in L_s}\frac{1}{N_s}\mathbb{E}[\langle\sigma_i\rangle_N]\,=\,
    \mathbb{E}[\langle m_s\rangle_N]\;,\\[2pt]
    &\mathbb{E}[\langle q_rq_s\rangle_N]\,=\,
    \sum_{(i,j)\in L_r\times L_s}\frac{\mathbb{E}[\langle\sigma_i\sigma_j\rangle_N^2]}{N_rN_s}\,=\,
    \sum_{(i,j)\in L_r\times L_s}\frac{\mathbb{E}[\langle\sigma_i\sigma_j\rangle_N]}{N_rN_s}\,=\,
    \mathbb{E}[\langle m_rm_s\rangle_N]\;,
\end{align}
where the expectations $\langle q_s\rangle_N$ and $\langle q_r q_s\rangle_N$ are taken with respect to the replicated Gibbs measure. As a consequence we have:
\begin{align}
\label{main_N_identity_quadratic}
    \mathbb{E}\Big\langle
    (\mathbf{q},\Delta\mathbf{q})
    \Big\rangle_N=\mathbb{E}\Big\langle
    (\mathbf{m},\Delta\mathbf{m})
    \Big\rangle_N\;.
\end{align}
Concerning the correlation inequalities on the Nishimori line \cite{Nishi_id_PC,contucci_morita_nishimori,morita2005} (see also Theorem 2.18 in \cite{contucci_giardina_2012} for a straightforward proof) we have that:
\begin{align}
    \label{corr_in_1}
    &\frac{\partial \bar{p}_N}{\partial h_r} \,=\,
    \frac{1}{2N}\sum_{i\in L_r}\mathbb{E}[1+\langle\sigma_i\rangle_N]\,=\,
    \frac{\alpha_r}{2}\,\left(1+\mathbb{E}\langle m_r\rangle_N\right)
    \,\geq 0\;,\\[2pt]
    \label{corr_in_2}
    &\frac{\partial^2 \bar{p}_N}{\partial h_r\partial h_s}\,=\,
    \frac{\alpha_r}{2}\,\frac{\partial \mathbb{E}\langle m_r\rangle_N}{\partial h_s}\,=\,
    \frac{1}{2N}\sum_{(i,j)\in L_r\times L_s}
    \mathbb{E}\left[\left(\langle\sigma_i\sigma_j\rangle_N-\langle\sigma_i\rangle_N\langle\sigma_j\rangle_N\right)^2\,\right]
    \,\geq 0\;.
\end{align}
Hence both the quenched pressure per particle and the magnetizations are non-decreasing with respect to each parameter $h_r$, $r=1,\dots,K\,$.

\subsection{One-body system on the Nishimori line}

It is useful to consider the following simple Hamiltonian on the Nishimori line, where only one-body interactions are taken into account:
\begin{align}
    \label{HAM_GAS}
    H_N^{(0)}(\sigma):=-\sum_{i=1}^N(z_i\,\sqrt{h}+h)\,\sigma_i\,,\quad z_i\iid\mathcal{N}(0,1)
\end{align}with $h>0$.
It is easy to show that the pressure of this model coincides with the function $\psi(h)$ defined by \eqref{psi_def}:
\begin{align}
    p_N^{(0)}:=\,\frac{1}{N}\mathbb{E}\log\sum_{\sigma\in\Sigma_N}e^{-H_N^{(0)}(\sigma)} \,=\, \psi(h) \;.
\end{align}
Since the Boltzmann-Gibbs average of a spin in the one body system equals $\langle \sigma_1\rangle_N^{(0)} = \tanh\big(z_1\sqrt{h}+h\big)\,$, the Nishimori identities entail the following identities:
\begin{equation} \label{tanh=tanh^2}
\E\tanh^{2n-1}\left(z\sqrt{h}+h\right) \,=\, \E\tanh^{2n}\left(z\sqrt{h}+h\right)
\end{equation}
for every $n\in\N,\,n\geq1$.
Starting from expression \eqref{psi_def} we are going to determine the sign of the first derivatives of $\psi\,$.
Gaussian integration by parts and identity \eqref{tanh=tanh^2} for $n=1$ show that
\begin{align} \label{increasing_psi_gas}
\psi'(h) \,=\, \frac{1}{2}\left(1+\mathbb{E}\tanh\left(z\sqrt{h}+h\right)\right) \,>0 \;.
\end{align}
Using again Gaussian integration by parts and identity \eqref{tanh=tanh^2} for $n=1,2$, one finds:
\begin{equation}  \label{convex_psi_gas}
\psi''(h) =\, \frac{1}{2}\,\E\left[\left(1-\tanh^2 \left(z\sqrt{h}+h\right)\right)^2\right] \,>0 \;.
\end{equation}
The sign of the third derivative can be obtained avoiding Gaussian integration by parts. Indeed by setting $y=z\sqrt{h}+h$, replacing $\frac{z}{2\sqrt h}+1=\frac{y+h}{2h}$ in the computations and using the identities \eqref{tanh=tanh^2} for $n=2,3$, one finds:
\begin{equation} \label{third_psi_gas}
\psi'''(h) =\, -\frac{1}{h}\,\E\left[\left(1-\tanh^2y\right)^2\,y\,\tanh y\right] -\,\E\left[\left(1-\tanh^2y\right)^2\,\tanh^2y\right] \,<0 \;.
\end{equation}
%
%
%
%
The convexity of $\psi$ will be crucial in the proof of Theorem \ref{main_theorem}.
In particular, we will use the following
\begin{lemma}\label{convex_f_gasliketerms}
The function
\begin{align}\label{gas-liketerms}
f(\mathbf{x}):=\sum_{r=1}^K\alpha_r\,\psi((M\mathbf{x})_r)
\end{align}
is convex for $\mathbf{x}$ such that $M\mathbf{x}\geq 0$ component-wise.
\end{lemma}

\begin{proof}
$\psi$ is convex on $\R_{\geq0}$ by equation \eqref{convex_psi_gas}. Then, using  the linearity of $(M\mathbf x)_r\,$, it is easy to verify that for any $\lambda\in[0,1]$ and $\mathbf{x}_1,\mathbf{x}_2\in A$ we have:
\begin{equation}
    f(\lambda \mathbf{x}_1+(1-\lambda)\mathbf{x}_2) \,\leq\,
    \lambda f(\mathbf{x}_1)+(1-\lambda) f(\mathbf{x}_2)\;.
\end{equation}
\end{proof}

In the proof of Theorem \ref{teor_uniq} we will use the following
\begin{lemma} \label{Fconcave}
Let $z$ be a standard Gaussian random variable. The function
\begin{equation} \label{F}
F(h) :=\, \E \tanh\left(z\sqrt{h}\,+h\right)
\end{equation}
is strictly positive, increasing and concave for $h>0$.
\end{lemma}
\begin{proof}
It follows immediately by equations \eqref{tanh=tanh^2}, \eqref{convex_psi_gas}, \eqref{third_psi_gas}, since $F=2\,\psi'-1\,$.
\end{proof}

\begin{remark}
As a consequence the function $F$ is invertible on $[0,\infty)\,$. Its inverse $F^{-1}$ is non negative and increasing on $[0,1)\,$. Moreover one has
\begin{align}\label{limit_F^-1}
    \lim_{x\to1^-}F^{-1}(x)=+\infty\,.
\end{align}
\end{remark}

\subsection{Interpolating model}
We now introduce an interpolating model that compares the original model with a one-body model with suitably tuned external field. 

\begin{definition}[Interpolating model] Let $t\in[0,1]$. The Hamiltonian of the interpolating model is:
\begin{multline}
    \label{H_interpolating_model}
    H_\sigma(t):=-\frac{\sqrt{1-t}}{\sqrt{2N}}\,\sum_{r,s=1}^K\sum_{(i,j)\in L_r\times L_s}J_{ij}^{rs}\sigma_i\sigma_j-(1-t)\,\frac{N}{2}\,(\mathbf{m},\Delta\mathbf{m})\;+\\-
    \sum_{r=1}^K\sum_{i\in L_r}\left(\sqrt{Q_{\epsilon,r}(t)}\,J_i^r+Q_{\epsilon,r}(t)\right)\sigma_i\;-\sum_{r=1}^K\sum_{i\in L_r}h^r_i\sigma_i-N(\hat{\alpha}\mathbf{h},\mathbf{m})
\end{multline}
with $J_{i}^{r}\iid\mathcal{N}\left(0,1\right)$ independent of all the other Gaussian random variables, and
\begin{align*}
    \mathbf{Q}_{\epsilon}(t):=\boldsymbol{\epsilon}+M\int_0^t\mathbf{q}_\epsilon(s)\,ds,\;\quad\epsilon_r\in[s_N,2s_N],\,s_N\propto N^{-\frac{1}{16K}} \;.
\end{align*}
Here $\mathbf{Q}_{\epsilon}=:(Q_{\epsilon,r})_{r=1,\dots,K}$, while $\mathbf{q}_\epsilon:=(q_{\epsilon,r})_{r=1,\dots,K}$ denotes a vector of $K$ non-negative functions that will be suitably chosen in the following.
\end{definition}

Now we can write the sum rule, which is contained in the following proposition. 
\begin{proposition}[Sum rule]The quenched pressure of the model rewrites as: 
\begin{multline}
    \label{Sum_rule_MSK_NL}
    \bar{p}_N(\mu,\mathbf{h})=\mathcal{O}(s_N)+\sum_{r=1}^K\,\alpha_r\,\psi(Q_{\epsilon,r}(1)+h_r)\;+\\
    +\int_0^1dt\,\left[\frac{
    (\mathbf{1}-\mathbf{q}_\epsilon(t),\Delta(\mathbf{1}-\mathbf{q}_\epsilon(t)))}{4}- \frac{(\mathbf{q}_\epsilon(t),\Delta\mathbf{q}_\epsilon(t))}{2}\right]+\frac{1}{4}\int_0^1dt\,R_\epsilon(t,\mu,\mathbf{h})\;,
\end{multline}
where the remainder is:
\begin{align}\label{remainder_definition}
    R_\epsilon(t,\mu,\mathbf{h})=\mathbb{E}\Big\langle
    (\mathbf{m}-\mathbf{q}_\epsilon(t),\Delta(\mathbf{m}-\mathbf{q}_\epsilon(t)))\Big\rangle_{N,t}^{(\epsilon)} \;.
\end{align}
\end{proposition}
\begin{proof}
We stress that the interpolating model is on the Nishimori line for any $t\in[0,1]$, as can be seen by direct inspection. This allows us to use the identities and inequalities for any value of the interpolating parameter.
See Proof of Proposition 2 in \cite{us} for the details.
\end{proof}

The tridiagonal form of $\Delta$ allows us to specialize the previous sum rule as follows:
\begin{equation}\label{sum_rule_DBM}\begin{split}
    \bar{p}_N(\mu,\mathbf{h})\,=\;& \mathcal{O}(s_N)\,+\,\sum_{r=1}^K\,\alpha_r\,\psi\left(Q_{\epsilon,r}(1)+h_r
    \right)\,+\\
    &+\sum_{r=1}^K\,\frac{\Delta_{r,r+1}}{2}\,\int_0^1dt\,\left[
    (1-q_{\epsilon,r}(t))(1-q_{\epsilon,r+1}(t))-2q_{\epsilon,r}(t)q_{\epsilon,r+1}(t)
    \right]\,+\\
    &+\sum_{r=1}^K\,\frac{\Delta_{r,r+1}}{2}\int_0^1 dt\,\mathbb{E}\big\langle(m_r-q_{\epsilon,r}(t))\,(m_{r+1}-q_{\epsilon,r+1}(t))\big\rangle_{N,t}^{(\epsilon)} \;,
\end{split}\end{equation}
or better, using the notation introduced for Theorem \ref{phase_transition_thm},
\begin{equation}\label{sum_rule_DBM_migliore}\begin{split}
    \bar{p}_N(\mu,\mathbf{h})\,=\;&\mathcal{O}(s_N)\,+\,\sum_{r=1}^K\,\alpha_r\,\psi\left(Q_{\epsilon,r}(1)+h_r
    \right)\,+\\
    &+\frac{1}{2}\int_0^1dt\,\left[
    (\mathbf{1}_o-\mathbf{q}_{\epsilon,o}(t),\Delta^{(oe)}(\mathbf{1}_e-\mathbf{q}_{\epsilon,e}(t)))-2(\mathbf{q}_{\epsilon,o}(t),\Delta^{(oe)}\mathbf{q}_{\epsilon,e}(t))
    \right]+\\
    &+\frac{1}{2}\int_0^1dt\;\mathbb{E}\big\langle(\mathbf{m}_o-\mathbf{q}_{\epsilon,o}(t),\Delta^{(oe)}(\mathbf{m}_{e}-\mathbf{q}_{\epsilon,e}(t)))\big\rangle_{N,t}^{(\epsilon)} \;,
\end{split}\end{equation}
where again the subscripts $o,e$ denote the odd or even components of a vector, $\mathbf{1}:=(1)_{r=1,\dots,K}\,$. 
We also denote
\begin{align}\label{interpolating_functions}
    \mathbf{Q}_{\epsilon,o}(t)=\boldsymbol{\epsilon}_o+M^{(oe)}\int_0^t\mathbf{q}_{\epsilon,e}(s)\,ds\,,\quad
    \mathbf{Q}_{\epsilon,e}(t)=\boldsymbol{\epsilon}_e+M^{(eo)}\int_0^t\mathbf{q}_{\epsilon,o}(s)\,ds\,.
\end{align}
The sum rules \eqref{sum_rule_DBM}, \eqref{sum_rule_DBM_migliore} motivate the definition of the variational pressure \eqref{p_var_DBM} that for future convenience can be rewritten as:
\begin{align}\label{p_var_dispari-pari}
    p_{var}(\mathbf{x};\mu,\mathbf{h})\,=\,\sum_{r=1}^K\,\alpha_r\,\psi\left((M\mathbf{x})_r+h_r
    \right)\,+\,\frac{(\mathbf{1}_o-\mathbf{x}_{o},\Delta^{(oe)}(\mathbf{1}_e-\mathbf{x}_{e}))}{2}-(\mathbf{x}_{o},\Delta^{(oe)}\mathbf{x}_{e})\,.
\end{align}

\begin{remark}\label{convexity_properties_pvar}
The variational function $p_{var}$ is convex in the even components $\mathbf{x}_e$ and the odd components $\mathbf{x}_o$ separately.
This is due to the fact that the two bilinear forms in \eqref{p_var_dispari-pari} have vanishing second derivatives w.r.t. pure odd or even components, while the terms containing $\psi$ are convex by Lemma \ref{convex_f_gasliketerms}.
\end{remark}

The sum rule exhibits a remainder (namely \eqref{remainder_definition}) to deal with. Let us first introduce the following 
\begin{definition}[Regularity of $\boldsymbol{\epsilon}\longmapsto\mathbf{Q}_{\epsilon}(\cdot)$]\label{def_regularity}
We will say that the map $\boldsymbol{\epsilon}\longmapsto\mathbf{Q}_{\epsilon}(\cdot)$ is \emph{regular} if
\begin{align}
    \text{det} \left(\frac{\partial\mathbf{Q_\epsilon}(t)}{\partial\boldsymbol{\epsilon}}\right)\geq 1\quad\forall t\in[0,1]
\end{align}
\end{definition}
This has to be combined with Liouville's formula, a standard analysis result that we report here for the reader's convenience.
\begin{lemma}[Liouville's formula]Consider two matrices whose elements depend on a real parameter: $\Phi(t),\,A(t)$. Suppose that $\Phi$ satisfies the Cauchy problem
\begin{align}
    \label{Liouville_edo}
    \begin{cases}
    \dot{\Phi}(t)=A(t)\,\Phi(t)\\
    \Phi(0)=\Phi_0
    \end{cases} \;.
\end{align}
Then:
\begin{align}\label{Liouville_formula}
    \text{\emph{det}}(\Phi(t))\,=\,\text{\emph{det}}(\Phi_0)\;\exp\left\{
    \int_0^tds\,\text{\emph{Tr}}(A(s))
    \right\}
\end{align}
\end{lemma}
Now, the remainder \eqref{remainder_definition} can be proved to concentrate under the regularity hypothesis, as stated in the following 
\begin{lemma}[Concentration]\label{concentration_lemma}
Suppose $\boldsymbol{\epsilon}\longmapsto\mathbf{Q}_\epsilon(\cdot)$ is a regular map. For every $r=1,\dots,K$ consider the quantity
\begin{align}
    \mathcal{L}_r:=\frac{1}{N_r}\sum_{i\in L_r}\left(\sigma_i+\frac{J_i^r\sigma_i}{2\sqrt{Q_{\epsilon,r}(t)}}\right),\;\quad J_i^r\iid\mathcal{N}(0,1)
\end{align}
and introduce the $\boldsymbol{\epsilon}$-average:
\begin{align}
    \mathbb{E}_{\boldsymbol{\epsilon}} [\cdot]=\prod_{r=1}^K\left(
    \frac{1}{s_N}\int_{s_N}^{2s_N}d\epsilon_r
    \right)(\cdot) \;.
\end{align}
We have:
\begin{align}
\mathbb{E}_{\boldsymbol{\epsilon}}\mathbb{E}\Big\langle\left(\mathcal{L}_r-\mathbb{E}\langle\mathcal{L}_r\rangle_{N,t}^{(\epsilon)}\right)^2\Big\rangle_{N,t}^{(\epsilon)}\longrightarrow 0 \quad\text{as }N\to\infty
\end{align}
and
\begin{align}
    \label{Concentration_1}
    \mathbb{E}\Big\langle\left(m_r-\mathbb{E}\langle m_r\rangle_{N,t}^{(\epsilon)}\right)^2\Big\rangle_{N,t}^{(\epsilon)} \,\leq\,
    4\,\mathbb{E}\Big\langle\left(\mathcal{L}_r-\mathbb{E}\langle\mathcal{L}_r\rangle_{N,t}^{(\epsilon)}\right)^2\Big\rangle_{N,t}^{(\epsilon)}
\end{align}
for every $r=1,\dots,K\,$.
Therefore the magnetization (or the overlap) concentrates in $\boldsymbol{\epsilon}$-average.
\end{lemma}
The proof is carried out by treating the thermal and disordered fluctuations of $\mathcal{L}_r$ separately. Actually, an estimate on the $L^2-$convergence speed of the random pressure to the quenched one as $N\to\infty$ is required and it will be given in the proofs section below. See Lemma 3 and Appendix A in \cite{us} for the details.
The role of $\boldsymbol{\epsilon}$ is that of a regularizing perturbation and it is crucial for the proof. Its introduction intuitively allows to avoid critical points where the limiting pressure presents singularities and concentration may not occur, thus helping us to select always the stable state of the system. Indeed, for vanishing external magnetic fields $\mathbf{h}=0$ and in absence of $\boldsymbol{\epsilon}$, the system remains stuck in a vanishing average magnetization state because of the resulting spin flip symmetry in the Hamiltonian. However, as stated in Theorem \ref{phase_transition_thm} in the appropriate range of parameters the latter is thermodynamically unstable, meaning that any arbitrarily small magnetic field would bring the magnetization to positive values.

\section{Proofs}\label{proofs_sect}

\subsection{Proof of Theorem \ref{main_theorem}}
The almost sure equality in \eqref{var_principle_DBM} is a standard result based on the following concentration inequality:
\begin{proposition}\label{tala_concentration}
There exists $C=C(\mu,\mathbf{h})>0$ such that for every $x>0$
\begin{align} \label{tala_bound}
    \mathbb{P}\left(\left|p_N-\bar{p}_N(\mu,\mathbf{h})\right|\geq x\right) \,\leq\, 2\exp\left(-\frac{Nx^2}{4C}\right) \;.
\end{align}
As a consequence
\begin{align} \label{self_av_press}
    \mathbb{E}[(p_N-\bar{p}_N(\mu,\mathbf{h}))^2] \,\leq\, \frac{8C}{N}  \;.
\end{align}
\end{proposition}
\begin{proof}
The random pressure $p_N$ is a Lipschitz function of the independent standard Gaussian variables $\hat J =(J_{ij}^{rs}/\sqrt{\mu_{rs}})_{i,j,r,s}\,$, $\hat h=(h_i^r/\sqrt{h_r})_{i,r}\,$. Indeed:
\begin{align}
   N^2\,|\nabla_{\!\hat J,\hat h}\,p_N|^2 \leq N\left(\frac{(\mathbf{1},\Delta\mathbf{1})}{2}+(\hat{\alpha}\mathbf{h},\mathbf{1})\right)\equiv\, C N
\end{align}
The inequality \eqref{tala_bound} then follows by a known concentration property of the Gaussian measure (see Theorem 1.3.4 in \cite{Tala_vol1}). A tail integration finally leads to \eqref{self_av_press}.
\end{proof}
Since the r.h.s. in \eqref{tala_bound} is summable the Borel-Cantellli Lemma guarantees almost sure convergence.
Now we move to the proof of the variational principle, \textit{i.e.} the second equality in \eqref{var_principle_DBM} which is going to be achieved through upper and lower bounds. For what follows, we neglect all the sub and superscripts in the Boltzmann-Gibbs averages, except for the $t$-dependence. 

\paragraph{Lower bound.} We select a path contained in $[0,1)^K$ by means of the following coupled ODEs
\begin{align}\label{path_constant_lower}
    &\dot{\mathbf{Q}}_{\epsilon,e}(t) \,=\, M^{(eo)}\,\mathbf{x}_o=:\mathbf{f}_e(t,\mathbf{Q}_\epsilon(t))\;,\quad\mathbf{Q}_{\epsilon,e}(0)=\boldsymbol{\epsilon}_e\\[4pt]
    \label{path_2_lower}
    &\dot{\mathbf{Q}}_{\epsilon,o}(t) \,=\,
    M^{(oe)}\,\mathbb{E}\langle \mathbf{m}_{e}\rangle_t=:\mathbf{f}_o(t,\mathbf{Q}_\epsilon(t))\;,\quad\mathbf{Q}_{\epsilon,o}(0)=\boldsymbol{\epsilon}_o\,,
\end{align}
where $\mathbf{f}(t,\mathbf{Q})$ is the velocity field of the ODE. The perturbation is here introduced as an initial condition in order to have the interpolating functions in the form \eqref{interpolating_functions}. 
Notice that $\mathbf{f}_e$ is constant, while $\mathbf{f}_o$ 
is a positive Lipschitz function of $\mathbf{Q}_\epsilon(t)\in(0,\infty)^K$
thanks to identity \eqref{corr_in_2} (where $N$ is fixed).
Therefore, by Cauchy-Lipschitz's theorem, the system of ODEs \eqref{path_constant_lower}-\eqref{path_2_lower} has a unique global solution $\mathbf{Q}_\epsilon(t)\,$, $t\in[0,1]\,$,
whose components are positive.

By \eqref{path_constant_lower}-\eqref{path_2_lower} we have $\Delta^{(eo)}\mathbf{q}_{\epsilon,o}(t)=\Delta^{(eo)}\mathbf{x}_o$ and $\Delta^{(oe)}\mathbf{q}_{\epsilon,e}(t)=\Delta^{(oe)}\mathbb{E}\langle \mathbf{m}_{e}\rangle_t\,$, hence:
\begin{multline}
    \int_0^1dt
    \left(\mathbf{1}_o-\mathbf{q}_{\epsilon,o}(t)\,,\,\Delta^{(oe)}(\mathbf{1}_e-\mathbf{q}_{\epsilon,e}(t))\right)=\\
    =
    \left(\mathbf{1}_o-\mathbf{x}_o\,,\,
    \Delta^{(oe)}\left(\mathbf{1}_e-\int_0^1dt\,\mathbb{E}\langle \mathbf{m}_{e}\rangle_t\right)\right) 
\end{multline}
and reasoning in a similar way for the other $t$-integrations appearing in the sum rule \eqref{sum_rule_DBM_migliore}
we obtain:
\begin{equation}\label{step1lower}\begin{split}
    \bar p_N \;&=\,
    \mathcal{O}(s_N)\,+\,p_{var}\!\left(
    \mathbf{x}_o\,,\int_0^1\!dt\,\mathbb{E}\langle \mathbf{m}_{e}\rangle_t
    \right)\,+\,\int_0^1 \!dt\,R_\epsilon(t)\,\geq\\
    &\geq\, \mathcal{O}(s_N)\,+\,\newinf_{\mathbf{x}_e}p_{var}\left(
    \mathbf{x}_o,\mathbf{x}_e\right)\,+\,\int_0^1\! dt\,R_\epsilon(t) \;,
\end{split}\end{equation}
where the reminder is
\begin{align}
R_\epsilon(t) \,=\,
\frac{1}{2}\,\mathbb{E}\big\langle\left( (\mathbf{m}_{o}-\mathbf{x}_{o})\,,\,\Delta^{(oe)}\,(\mathbf{m}_{e}-\mathbb{E}\langle \mathbf{m}_{e}\rangle_t)\right)\big\rangle_{t} \;.
\end{align}
Using Cauchy-Schwartz's inequality,
\begin{align}\label{bound_remainder}
|R_\epsilon(t)| \,\leq\, \frac{1}{2}\,\left\Vert\mu^{(oe)}\right\Vert\, \E^{1/2}\langle|\hat{\alpha}^{(oo)}(\mathbf{m}_{o}-\mathbf{x}_{o})|^2\rangle_{t}\;
\E^{1/2}\langle|\hat{\alpha}^{(ee)}(\mathbf{m}_{e}-\mathbb{E}\langle \mathbf{m}_{e})\rangle_t|^2\rangle_{t} \;,
\end{align}
thus, provided that the map $\boldsymbol{\epsilon}\mapsto\mathbf{Q}_\epsilon(t)$
is regular,
the remainder $R_\epsilon(t)$ vanishes in $\boldsymbol{\epsilon}$-average as $N\to\infty$ by Lemma \ref{concentration_lemma}.
To show that $\mathbf{Q}_\epsilon$ is regular we introduce the following matrix fields:
\begin{align}
    \Phi_{\epsilon}(t):= \,\frac{\partial\, \mathbf{Q}_{\epsilon}(t)}{\partial\boldsymbol{\epsilon}} \;,\quad
    A_\epsilon(t) :=\, \frac{\partial\, \mathbf{f}(t,\mathbf{Q}_\epsilon(t))}{\partial\,\mathbf{Q}_\epsilon(t)} 
\end{align}
Applying the chain rule we have:
\begin{align}
    \dot{\Phi}_\epsilon(t) =\, \frac{\partial \,\dot{\mathbf{Q}}_\epsilon(t)}{\partial\boldsymbol\epsilon} \,=\, A_\epsilon(t)\;\Phi_\epsilon(t) \;,\quad  \Phi_\epsilon(0)=\mathbbm{1}  \;,
\end{align}
hence, by Liouville's formula \eqref{Liouville_formula} the Jacobian det$(\Phi_\epsilon(t))$ is
\begin{align}\label{liouville_in_teorema}
    \text{det}\Big(\frac{\partial \mathbf{Q}_\epsilon}{\partial\boldsymbol{\epsilon}}(t)\Big) \,=\,
    \exp\left\{
    \int_0^tds\;\text{Tr}\big( A_\epsilon(s) \big)
    \right\}\,.
\end{align}
%
%
%
%
Now, using equations \eqref{path_constant_lower}-\eqref{path_2_lower} one can compute:
\begin{equation}\label{trace} \begin{split}
\text{Tr} \big( A_\epsilon(t) \big) \,&=\,
\sum_{r=1}^K \big(A_\epsilon(t)\big)_{r,r} \,=\,
\sum_{r\text{ odd}} \frac{\partial\big( M^{(oe)}\,\E\langle \mathbf{m}_e\rangle_t\big)_r}{\partial\, Q_{\epsilon,r}(t)} \,=\\[4pt]
&=\,\sum_{r\text{ odd}}\sum_{r'\text{even}} M_{rr'}\,\frac{\partial\,\E\langle m_{r'}\rangle_{t}}{\partial\,Q_{\epsilon,r}(t)}
\;\geq0
\end{split}\end{equation}
where non-negativity is a consequence of the correlation inequality \eqref{corr_in_2}, since $Q_{\epsilon,r}(t)$ can be seen as the variance of an external field on the Nishimori
line in the interpolating Hamiltonian \eqref{H_interpolating_model}.
Combining \eqref{liouville_in_teorema} and \eqref{trace}, it follows that $\mathbf{Q}_\epsilon$ is regular, as desired.

Now, averaging on $\boldsymbol{\epsilon}$ and tanking the $\liminf_{N\to\infty}$ in inequality \eqref{step1lower} we have
\begin{align}
    \liminf_{N\to\infty}\bar{p}_N \,\geq \, \newinf_{\mathbf{x}_e}p_{var}\left(
    \mathbf{x}_o,\mathbf{x}_e\right)+\liminf_{N\to\infty}\mathbb{E}_{\boldsymbol{\epsilon}}\int_0^1\! dt\,R_\epsilon(t)\,.
\end{align}
The last term vanishes by Fubini's theorem, dominated convergence and Lemma \ref{concentration_lemma}. Finally, optimizing w.r.t. $\mathbf{x}_o$ we get:
\begin{align}
    \liminf_{N\to\infty}\bar{p}_N\geq \sup_{\mathbf{x}_o}\newinf_{\mathbf{x}_e}p_{var}\left(
    \mathbf{x}_o,\mathbf{x}_e\right)\,.
\end{align}

\paragraph{Upper bound.} Now, we set
\begin{align}\label{choice_path_upper_odd}
    &\dot{\mathbf{Q}}_{\epsilon,e}(t)=M^{(eo)}F(M^{(oe)}\mathbb{E}\langle \mathbf{m}_{e}\rangle_t+\mathbf{h}_o)\,,\quad\mathbf{Q}_{\epsilon,e}(0)=\boldsymbol{\epsilon}_e\\
    \label{choice_path_upper_even}
    &\dot{\mathbf{Q}}_{\epsilon,o}(t)=M^{(oe)}\mathbb{E}\langle \mathbf{m}_{e}\rangle_t\,,\quad\mathbf{Q}_{\epsilon,o}(0)=\boldsymbol{\epsilon}_o\,.
\end{align}
In \eqref{choice_path_upper_odd} the application of $F$, defined in \eqref{F}, to the vector $M^{(oe)}\mathbb{E}\langle \mathbf{m}_{e}\rangle_t+\mathbf{h}_o$ has to be understood as component-wise. For future convenience let us set
\begin{align}\label{matchalD}
    \mathcal{D}(\mathbf{x},\mathbf{h}):=\text{diag}\left\{F'\left((M\mathbf{x})_r+h_r\right)\right\}_{r=1,\dots,K}.
\end{align}
With a slight abuse of notation we will stress the dependence of $\mathcal{D}^{(oo)}(\mathbf{x},\mathbf{h})$ and $\mathcal{D}^{(ee)}(\mathbf{x},\mathbf{h})$ on the even and odd components of $\mathbf{x}$ respectively as follows
\begin{align}
    \mathcal{D}^{(oo)}(\mathbf{x},\mathbf{h})\equiv \mathcal{D}^{(oo)}(\mathbf{x}_e,\mathbf{h})\,,\quad\mathcal{D}^{(ee)}(\mathbf{x},\mathbf{h})\equiv \mathcal{D}^{(ee)}(\mathbf{x}_o,\mathbf{h})\,.
\end{align}
$M^{(eo)}F(M^{(oe)}\mathbb{E}\langle \mathbf{m}_{e}\rangle_t+\mathbf{h}_o)$ is a positive function of $\mathbf{Q}_{\epsilon}(t)$ with bounded derivatives for fixed $N$ thanks to Lemma \ref{Fconcave}, indeed
\begin{align}
    \frac{\partial }{\partial Q_{\epsilon,r}}F(M^{(oe)}\mathbb{E}\langle \mathbf{m}_{e}\rangle_t+\mathbf{h}_o)=\mathcal{D}(\mathbb{E}\langle \mathbf{m}_{e}\rangle_t,\mathbf{h})^{(oo)}M^{(oe)}\frac{\partial \mathbb{E}\langle \mathbf{m}_{e}\rangle_t}{\partial Q_{\epsilon,r}}\,,
\end{align}This ensures the existence of a unique global solution over $[0,1]$ to the system of ODEs \eqref{choice_path_upper_odd}-\eqref{choice_path_upper_even}. Moreover, the latter implies also that the map $\boldsymbol{\epsilon}\longmapsto \mathbf{Q}_{\epsilon}(\cdot)$ is still regular, because  $F'$ is positive as proved in Lemma \ref{Fconcave} and $\frac{\partial \mathbb{E}\langle \mathbf{m}_{e}\rangle_t}{\partial Q_{\epsilon,r}}\geq0$ thanks again to \eqref{corr_in_2}. This guarantees the positivity of the trace in \eqref{liouville_in_teorema} and forces the vanishing of the remainder $R_\epsilon$ in $\boldsymbol{\epsilon}$-average by Lemma \ref{concentration_lemma}.
Using Jensen's inequality, by the convexity of $\psi$ we have
\begin{align}
    \sum_{r=1}^K\alpha_r\psi\left(\left(M \int_0^1\mathbf{q}_\epsilon(t)\,dt+\mathbf{h}\right)_{\!r}\right)\leq\sum_{r=1}^K\alpha_r\int_0^1\psi\left(\left(M \mathbf{q}_\epsilon(t)+\mathbf{h}\right)_{\!r}\right)\,dt\,
\end{align}and inserting it into the sum rule \eqref{sum_rule_DBM_migliore} we get
\begin{multline}\label{step_inf_1}
    \bar{p}_N\leq \mathcal{O}(s_N)+\int_0^1dt\,p_{var}(\mathbf{F}_{\epsilon,o}(t),\mathbb{E}\langle \mathbf{m}_{e}\rangle_t)+\int_0^1R_\epsilon(t)\,dt=\\=
    \mathcal{O}(s_N)+\int_0^1dt\,\newinf_{\mathbf{x}_e}p_{var}(\mathbf{F}_{\epsilon,o}(t),\mathbf{x}_e)+\int_0^1R_\epsilon(t)\,dt\,,
\end{multline}
where $\mathbf{F}_{\epsilon,o}(t):=F(M^{(oe)}\mathbb{E}\langle \mathbf{m}_{e}\rangle_t+\mathbf{h}_o)$ for brevity.
As far as the last equality is concerned, we used the following:
\begin{align}\label{step_inf_2}
    \newinf_{\mathbf{x}_e}p_{var}(\mathbf{F}_{\epsilon,o}(t),\mathbf{x}_e)=p_{var}(\mathbf{F}_{\epsilon,o}(t),\mathbb{E}\langle \mathbf{m}_{e}\rangle_t)\;.
\end{align}
This is a consequence of the convexity of $p_{var}$ in $\mathbf{x}_e$ (see Remark \ref{convexity_properties_pvar}). In fact, a computation of the gradient of $p_{var}$ w.r.t. $\mathbf{x}_e$ evaluated at $\mathbb{E}\langle \mathbf{m}_{e}\rangle_t$ yields:
\begin{multline}
    \left.\frac{\partial p_{var}}{\partial \mathbf{x}_e}(\mathbf{F}_{\epsilon,o}(t),\mathbf{x}_e)\right|_{\mathbb{E}\langle \mathbf{m}_{e}\rangle_t}=\frac{\Delta^{(eo)}}{2}[\mathbf{1}_o+\mathbf{F}_{\epsilon,o}(t)]+\\+\frac{\Delta^{(eo)}}{2}[-\mathbf{1}_o+\mathbf{F}_{\epsilon,o}(t)]-\Delta^{(eo)}\mathbf{F}_{\epsilon,o}(t)=0\,,
\end{multline}
where we explicitly notice that the first term comes from the derivative of $\psi$ \eqref{increasing_psi_gas}.
Then, taking the sup of $p_{var}$ over the odd components and the $\boldsymbol{\epsilon}$-average we get:
\begin{align}
    \bar{p}_N\leq \mathcal{O}(s_N)+\sup_{\mathbf{x}_o}\newinf_{\mathbf{x}_e}p_{var}\left(
    \mathbf{x}_o,\mathbf{x}_e\right)+\mathbb{E}_{\boldsymbol{\epsilon}}\int_0^1R_\epsilon(t)\,dt\,.
\end{align}Applying Lemma \ref{concentration_lemma}, Fubini's theorem and dominated convergence the two bounds match after sending $N\to\infty$.

\paragraph{Proof of \eqref{exchange_lim}.}
Equations \eqref{corr_in_1} and \eqref{corr_in_2} imply that the quenched pressure is convex in each $h_r$. Hence it is possible to exchange the derivative w.r.t. $h_r$ in \eqref{corr_in_1} with the $N\to\infty$ limit where $\bar{\mathbf{x}}$ is differentiable in $h_r$ (see Lemma IV.6.3 in \cite{ellis_book_largedev}). Since for invertible $\Delta$ the optimal order parameter must be a critical point of $p_{var}$ (see Proposition \ref{proposition_sup_pi} below) by \eqref{increasing_psi_gas} and \eqref{p_var_DBM} we have that:
\begin{multline}
    \lim_{N\to\infty}\frac{\partial \bar{p}_N}{\partial h_r}=
    \frac{\partial }{\partial h_r}p_{var}(\bar{\mathbf{x}}(M,\mathbf{h});\mu,\mathbf{h})=\left.\frac{\partial p_{var}}{\partial \mathbf{x}}\right|_{\bar{\mathbf{x}}(M,\mathbf{h})}\frac{\partial \bar{\mathbf{x}}(M,\mathbf{h})}{\partial h_r}+\frac{\partial p_{var}}{\partial h_r}=\frac{\partial p_{var}}{\partial h_r}=\\=
    \alpha_r\,\psi'((M\mathbf{\bar{x}}(M,\mathbf{h}))_r+h_r)=\frac{\alpha_r}{2}\left[1+\E_z\tanh\left(z\,\sqrt{(M\mathbf{\bar{x}})_r+h_r\,}\,+\,(M\mathbf{\bar{x}})_r+h_r\right)\right]=\\=\frac{\alpha_r}{2}\left[1+\bar{x}_r\right]\,.
\end{multline}
A comparison with $\eqref{corr_in_1}$ and the Nishimori identity \eqref{N_identity_1} lead to the identification:
\begin{align}
    \lim_{N\to\infty}\E \langle q_r \rangle_N \,=\, \lim_{N\to\infty}\E \langle m_r \rangle_N \;=\, \bar{x}_r \,.
\end{align}

\begin{remark}\label{extension_inf}
Assume for now that $K$ is even. Observe that the entire proof could have been carried out also by computing all the $\inf_{\mathbf{x}_e}$ over the convex set:
\begin{align}
    A:=\{\mathbf{x}_e\,|\,M^{(oe)}\mathbf{x}_e+\mathbf{h}_o\geq 0\text{ component-wise}\}\supseteq [0,1)^{K/2}\,,
\end{align}
on which all the functions involved are still real and well defined. This freedom is essentially due to the convexity of $p_{var}$ in $\mathbf{x}_e$. Indeed, $p_{var}$ has always a critical point in the domain $A$ for any fixed $\mathbf{x}_o\in[0,1)^{K/2}$, that must coincide with its minimum point by convexity as can be seen by direct inspection
\begin{align}\label{critical_x_even}
    \left.\frac{\partial p_{var}}{\partial \mathbf{x}_e}\right|_{\bar{\mathbf{x}}_e}=\frac{\Delta^{(eo)}}{2}\left[-\mathbf{x}_o+F(M^{(oe)}\bar{\mathbf{x}}_e+\mathbf{h}_o)\right]=0\;\Leftrightarrow\; 
    \bar{\mathbf{x}}_e=[M^{(oe)}]^{-1}(F^{-1}(\mathbf{x}_o)-\mathbf{h}_o)\,.
\end{align}The inequality \eqref{step1lower}, that leads to the lower bound, clearly holds also for $\mathbf{x}_e\in A\supseteq [0,1)^{K/2}$. The validity of \eqref{step_inf_2} is less trivial and is due to the special choice $\mathbf{F}_{\epsilon,o}(t)$. In this case in fact, the critical point falls inside $[0,1)^{K/2}$ and this lets us extend the domain of $\mathbf{x}_e$ to $A$ without any loss of generality thanks to the mentioned convexity in $\mathbf{x}_e$. We will see later that even with this extension the point that realizes the $\sup \inf$ lies inside the cube $[0,1)^K$.
\end{remark}

\subsection{Proof of Theorem \ref{phase_transition_thm}}
For this proof we rely on Remark \ref{extension_inf}, this will ease our computations. Let us write the gradient of \eqref{p_var_DBM}
\begin{multline}\label{gradient_of_pvar_DBM}
    \frac{\partial p_{var}(\mathbf{x};\mu,\mathbf{h})}{\partial x_r}=\left(\frac{\Delta}{2}(-\mathbf{x}+F(M\mathbf{x}+\mathbf{h}))\right)_{\!r}=\\=\frac{\Delta_{r,r+1}}{2}\left[-x_{r+1}+F((M\mathbf{x})_{r+1}+h_{r+1})\right]+\frac{\Delta_{r,r-1}}{2}\left[-x_{r-1}+F((M\mathbf{x})_{r-1}+h_{r-1})\right]\,.
\end{multline}where we have used \eqref{tanh=tanh^2}. In absence of external magnetic field ($\mathbf{h}=0$) $\mathbf{x}=0$ is a critical point for $p_{var}$, namely a solution to the consistency equation obtained by equating \eqref{gradient_of_pvar_DBM} to $0$.

First of all, by Remark \ref{convexity_properties_pvar} and Remark \ref{extension_inf} we infer that the optimization w.r.t. the even components $\mathbf{x}_e$ is always stable, in the sense that there is always one optimizer once the odd components $\mathbf{x}_o$ are fixed and it belongs to $A$. Define now the auxiliary function:
\begin{align}\label{def_pi}
    \pi(\mathbf{x}_o;\mu,\mathbf{h}) :=\inf_{\mathbf{x}_e\in A}p_{var}(\mathbf{x}_o,\mathbf{x}_e;\mu,\mathbf{h})=p_{var}(\mathbf{x}_o,\bar{\mathbf{x}}_e;\mu,\mathbf{h})\,,
\end{align}
with $\bar{\mathbf{x}}_e$ defined in \eqref{critical_x_even}.
The following proposition investigates the possibility to have boundary solutions to the variational problem.
\begin{proposition}\label{proposition_sup_pi}
Let $K$ be even. The points $\mathbf{x}_o$ at which the $\sup_{\mathbf{x}_o}\pi(\mathbf{x}_o;\mu,\mathbf{h})$ is attained fulfill the consistency equation:
\begin{align}\label{consist_eq_implicit}
    \bar{\mathbf{x}}_e=F(M^{(eo)}\mathbf{x}_o+\mathbf{h}_e)\,.
\end{align}
As a consequence the necessary condition for $\mathbf{x}$ to realize the $\sup_{\mathbf{x}_o}\newinf_{\mathbf{x}_e}p_{var}(\mathbf{x}_o,\mathbf{x}_e;\mu,\mathbf{h})$ is to be a critical point, namely to satisfy \eqref{consist_eq_implicit}.
\end{proposition}
\begin{proof}
Using \eqref{critical_x_even}, the gradient of $\pi$ is:
\begin{multline}\label{gradient_pi}
    \frac{\partial \pi(\mathbf{x}_o;\mu,\mathbf{h})}{\partial \mathbf{x}_o}=\frac{\partial p_{var}}{\partial \mathbf{x}_o}+\left.\frac{\partial p_{var}}{\partial \mathbf{x}_e}\right|_{\bar{\mathbf{x}}_e}\frac{\partial\bar{\mathbf{x}}_e}{\partial\mathbf
    {x}_o}=\frac{\Delta^{(oe)}}{2}\left[-\bar{\mathbf{x}}_e+F(M^{(eo)}\mathbf{x}_o+\mathbf{h}_e)\right]=\\=
    \frac{\hat{\alpha}^{(oo)}}{2}\left[-F^{-1}(\mathbf{x}_o)+\mathbf{h}_o+M^{(oe)}F(M^{(eo)}\mathbf{x}_o+\mathbf{h}_e)\right]\,.
\end{multline}We start by considering the case $\mathbf{h}=0$. 
One can immediately rule out the possibility that the $\sup$ is attained at the right border, \emph{i.e.} $x_{2l-1}\to 1^-$ for some $l$, because thanks to \eqref{limit_F^-1} $\partial_{x_{2l-1}}\pi\to-\infty$. Then, the necessary condition for a point $\mathbf{x}_o\in[0,1)^{K/2}$ to realize the sup is that:
\begin{align}
    -F^{-1}(\mathbf{x}_o)+M^{(oe)}F(M^{(eo)}\mathbf{x}_o)\leq0\,,
\end{align}component-wise, where equality holds for those components for which $x_{2l-1}>0$. If we set $M_{0,1}=M_{K,K+1}=0$, the generic $2l-1$ component of the previous is given by
\begin{multline}
    -F^{-1}(x_{2l-1})+M_{2l-1,2l-2}F(M_{2l-2,2l-3}x_{2l-3}+M_{2l-2,2l-1}x_{2l-1})+\\+M_{2l-1,2l}F(M_{2l,2l-1}x_{2l-1}+M_{2l,2l+1}x_{2l+1})
\end{multline}whence we understand that if $x_{2l-1}=0$ the only chance for the previous to be non positive is to have also $x_{2l-3}=x_{2l+1}=0$ because $F$ is positive and monotonic. On the contrary, if $x_{2l-1}>0$ first the corresponding gradient component must vanish; second by looking at the $2l+1$ component for instance
\begin{multline}
    -F^{-1}(x_{2l+1})+M_{2l+1,2l+2}F(M_{2l+2,2l+3}x_{2l+3}+M_{2l+2,2l+1}x_{2l+1})+\\+M_{2l+1,2l}F(M_{2l,2l+1}x_{2l+1}+M_{2l,2l-1}x_{2l-1})
\end{multline}we see that the last term is strictly positive. Necessarily, $x_{2l+1}$ must be strictly positive too with the corresponding gradient component that vanishes, and so on. Similar considerations hold for $x_{2l-3}$. Finally, iterating these arguments, we infer that the supremum is attained at a point $\mathbf{x}_o$ such that:
\begin{align}
    \mathbf{x}_o=0\quad\text{ or }\quad \bar{\mathbf{x}}_e=F(M^{(eo)}\mathbf{x}_o)\,.
\end{align}The first in particular implies that also $\bar{\mathbf{x}}_e=0$\,. In both cases we can say that \eqref{consist_eq_implicit} is satisfied. 

When any $h_r$ is strictly positive it is immediate to see that there is a component of \eqref{gradient_pi} with a positive contribution, the corresponding component of $\mathbf{x}_o$ must then be positive. Therefore one iterates the same arguments as above obtaining again \eqref{consist_eq_implicit}. In any case, by \eqref{critical_x_even} the $\sup\inf$ is attained at critical points of $p_{var}$. 
\end{proof}

Using property \eqref{convex_psi_gas}, the Jacobian matrix of $F(M\mathbf{x}+\mathbf{h})$ is
\begin{align}\label{Jacobian_T}
    DF(M\mathbf{x}+\mathbf{h})=\mathcal{D}(\mathbf{x},\mathbf{h})M\,.
\end{align}
$\mathcal{D}(\mathbf{x},\mathbf{h})$, defined in \eqref{matchalD}, is diagonal, positive definite, invertible and its spectral radius is bounded by $1$. 
From \eqref{gradient_pi}, an application of the Inverse Function Theorem leads to the Hessian matrix 
\begin{multline}
    \mathscr{H}_{\mathbf{x}_o}\pi=
    \frac{\Delta^{(oe)}}{2}\left[-
    \frac{\partial\bar{\mathbf{x}}_e}{\partial\mathbf{x}_o}+
    \frac{\partial }{\partial\mathbf{x}_o}F(M^{(eo)}\mathbf{x}_o+\mathbf{h_e})
    \right]=\\\frac{\Delta^{(oe)}}{2}\left[-
    [M^{(oe)}]^{-1}[\mathcal{D}^{(oo)}(\bar{\mathbf{x}}_e,\mathbf{h})]^{-1}+
    \mathcal{D}^{(ee)}(\mathbf{x}_o,\mathbf{h})M^{(eo)}
    \right]\,.
\end{multline}
Thanks to the peculiar tridiagonal form of $M$ we also have that 
\begin{align}\label{relazione_cattiva}
    [\mathcal{D}(\mathbf{x},\mathbf{h})M]^{(oe)}[\mathcal{D}(\mathbf{x},\mathbf{h})M]^{(eo)}=[(\mathcal{D}(\mathbf{x},\mathbf{h})M)^{2}]^{(oo)}\,,
\end{align}
from which by a simple rearrangement we can write the Hessian in its final form
\begin{multline}\label{Hessian_pi}
    \mathscr{H}_{\mathbf{x}_o}\pi
    =\frac{\hat{\alpha}^{(oo)}[\mathcal{D}(\bar{\mathbf{x}}_e,\mathbf{h})^{(oo)}]^{-1}}{2}\left[-\mathbbm{1}+(\mathcal{D}(\mathbf{x},\mathbf{h})M)^2\right]^{(oo)}=\\=
    \frac{[\alpha^{(oo)}]^{1/2}[\mathcal{D}^{(oo)}]^{-1/2}}{2}\left[-\mathbbm{1}+S^{(oo)}\right][\alpha^{(oo)}]^{1/2}[\mathcal{D}^{(oo)}]^{-1/2}
    \,
\end{multline}with
\begin{align}
    S^{(oo)}:=[\mathcal{D}^{(oo)}]^{1/2}[\hat{\alpha}^{(oo)}]^{-1/2}
    \Delta^{(oe)}\mathcal{D}^{(ee)}[\hat{\alpha}^{(ee)}]^{-1}\Delta^{(eo)}
    [\hat{\alpha}^{(oo)}]^{-1/2}[\mathcal{D}^{(oo)}]^{1/2}
\end{align}
where for brevity we have neglected all the dependencies after the second equality in \eqref{Hessian_pi} and used \eqref{relazione_cattiva}. \eqref{Hessian_pi} uses only symmetric matrices in order to make manifest the global sign of the Hessian.
It remains to show that the spectral radius of $S^{(oo)}$ is controlled by that of $[M^{2}]^{(oo)}$. $S^{(oo)}$ is symmetric because $\Delta^{(oe)}=[\Delta^{(eo)}]^{T}$, thus its spectral radius coincides with the matrix norm induced by the Euclidean scalar product.
Then by norms sub-multiplicativity and matrix similarity one easily gets
\begin{multline}
    \rho\left(S^{(oo)}\right)\leq \rho\left(\mathcal{D}^{(oo)}\right) \rho\left([\hat{\alpha}^{(oo)}]^{-1/2}
    \Delta^{(oe)}\mathcal{D}^{(ee)}[\hat{\alpha}^{(ee)}]^{-1}\Delta^{(eo)}
    [\hat{\alpha}^{(oo)}]^{-1/2}\right)\leq\\
    \leq\rho\left(M^{(oe)}\mathcal{D}^{(ee)}M^{(eo)}\right)=
    \rho\left(\mathcal{D}^{(ee)}M^{(eo)}M^{(oe)}\right)=\\=\rho\left([\mathcal{D}^{(ee)}]^{1/2}[\hat{\alpha}^{(ee)}]^{-1/2}
    \Delta^{(eo)}\hat{\alpha}^{(oo)-1}\Delta^{(oe)}
    [\hat{\alpha}^{(ee)}]^{-1/2}[\mathcal{D}^{(ee)}]^{1/2}\right)\,.
\end{multline}
Iterating the same arguments we get to
\begin{align}
    \rho\left(S^{(oo)}\right)\leq\rho\left(
    M^{(eo)}M^{(oe)}
    \right)=\rho\left(M^{(oe)}M^{(eo)}\right)<1\,,
\end{align}where the last equality follows again by matrix similarity. The previous implies that the matrix $\left[-\mathbbm{1}+S\right]^{(oo)}$ in \eqref{Hessian_pi} is negative definite, making $\mathscr{H}_{\mathbf{x}_o}\pi$ negative definite too, and hence $\pi$ is concave under the hypothesis $\rho([M^2]^{(oo)})<1$. In turn, this ensures uniqueness of the solution to the consistency equation \eqref{consist_eq_implicit} and to the variational problem \eqref{var_principle_DBM}. In particular when $\mathbf{h}=0$, $\mathbf{x}=0$ is the unique solution.

Conversely, for $\mathbf{h}=0$ and $\rho([M^2]^{(oo)})>1$ the Hessian has at least one positive eigenvalue at the origin $\mathbf{x}_o=0$, but this is in general not enough to ensure $\mathbf{x}_o=0$ does not realize the $\sup$ anymore. One has to check that there is a direction of increment of $\pi$ that intersects the cube $[0,1)^{K/2}$, otherwise the system could remain stuck on the border at $\mathbf{x}_o=0$ due to the positivity constraints on the variables.

It is easy to see that $[M^{2}]^{(oo)}$ is irreducible, because its associated graph is strongly connected, and it has non negative entries. Hence, by Perron-Frobenius Theorem the eigenvector $\mathbf{v}$ relative to the largest eigenvalue $\rho([M^2]^{(oo)})$ is component-wise strictly positive, thus pointing inside the cube, and by a Taylor expansion around $\mathbf{x}_o=0$ we have:
\begin{align}
    \pi(\epsilon\mathbf{v};\mu,0)-\pi(0;\mu,0)=\frac{\epsilon^2}{2}\left(\mathbf{v},\frac{\hat{\alpha}^{(oo)}}{2}\mathbf{v}\right)\left[-1+\rho([M^2]^{(oo)})\right]+o(\epsilon^2)
\end{align}that is positive form small enough $\epsilon>0$. Finally, by Proposition \ref{proposition_sup_pi} the solution shifts in favour of a point $\mathbf{x}=\bar{\mathbf{x}}(M)\in(0,1)^K$.

\subsection{Proof of Proposition \ref{prop:rhobound}}
Proposition \ref{prop:rhobound} relies on an algebraic lemma, which we write here for convenience. Its proof can be found in \cite{noi_deep2} (see Lemma 1 therein).
\begin{lemma} \label{lem:diseq}
Let $P\geq2$, $x_1,\dots,x_P\geq0$ and $b_1,\dots,b_{P-1}\geq0\,$. Set $S\equiv \sum_{p=1}^Px_p$ and $B \equiv \max_{p=1,\dots,P-1} b_p\,$. Then:
\be \label{diseq}
\sum_{p=1}^{P-1}\, b_p\,x_p\,x_{p+1} \,\leq\, \frac{B\,S^2}{4} \;.
\ee
Moreover we have equality in \eqref{diseq} if and only if one of the following conditions is verified:
\begin{itemize}
\item[a)] there exists $p^*\in\{1,\dots,P-1\}$ such that
\be \label{onehalf2}
x_{p^*} \,=\, x_{p^*+1} \,=\, \frac{S}{2} \quad,\quad b_{p^*} = B \;;
\ee
\item[b)] there exists $p^*\in\{2,\dots,P-1\}$ such that
\be \label{onehalf1}
x_{p^*} \,=\, x_{p^*-1}+x_{p^*+1} \,=\, \frac{S}{2} \quad,\quad b_{p^*-1} = b_{p^*} = B \;.
\ee
\end{itemize}
\end{lemma}

\begin{proof}[Proof of Proposition \ref{prop:rhobound}]
Denote by $\rho$ the spectral radius of the matrix $[M^2]^{(oo)}$. We have
\be \label{rho-norma}
\rho \,\leq\, \left\Vert [M^2]^{(oo)} \right\Vert_\infty \;.
\ee
As $[M^2]^{(oo)}$ is a tridiagonal matrix, its $\infty$-norm can be easily computed leading to
\be \label{norma}
\left\Vert [M^2]^{(oo)} \right\Vert_\infty \,=\,
\max_{r}\,\sum_{s}\, (M^2)_{2r-1,2s-1}
\,=\, \max_{r} \sum_{p=2r-3}^{2r} b_{p}^{(r)}\, \alpha_p\,\alpha_{p+1}\,
\leq\, \frac{{\widehat\mu}^2}{4} \;,\ee
where we set ${\widehat\mu}^2\equiv\max_r\mu_{r,r+1}^2$ and for every $r,p$
\be \begin{split}
b_{p}^{(r)} \,\equiv\,\; & \delta_{p,\,2r-3}\;\mu_{2r-3,\,2r-2}\;\mu_{2r-2,\,2r-1} \,+\,
\delta_{p,\,2r-2}\;\mu_{2r-2,\,2r-1}^2 \,+\\
& +\, \delta_{p,\,2r-1}\;\mu_{2r-1,\,2r}^2 \,+\,
\delta_{p,\,2r}\;\mu_{2r-1,\,2r}\;\mu_{2r,\,2r+1}  \;.
\end{split} \ee
For convenience we set $\alpha_{p}\equiv0$ for $p\notin\{1,\dots,K\}\,$ and $\mu_{p,p+1}\equiv0$ for $p\notin\{1,\dots,K-1\}$.
The last inequality in \eqref{norma} follows by Lemma \ref{lem:diseq}, since $b_p^{(r)}\leq {\widehat\mu}^2$ and  $\sum_p \alpha_p=1\,$.
Therefore $\rho\leq\frac{{\widehat\mu}^2}{4}$ combining inequalities \eqref{rho-norma}, \eqref{norma}.

Now assume that $\rho=\frac{{\widehat\mu}^2}{4}$. In particular the inequality in \eqref{norma} must be saturated, hence there exists $r$ such that
\be 
\sum_{p=2r-3}^{2r}b_{p}^{(r)}\,\alpha_{p}\,\alpha_{p+1} \,=\, \frac{{\widehat\mu}^2}{4} \;.
\ee
Then by Lemma \ref{lem:diseq}, condition \eqref{bestlambda1} or \eqref{bestlambda2} must be verified.

Vice-versa assume that condition \eqref{bestlambda1} or \eqref{bestlambda2} holds true. In this case notice that many of the $\alpha_r$'s are zero, since $\sum_{r=1}^K\alpha_r=1\,$. 
Thus the matrix $[M^2]^{(oo)}$ notably simplifies and one can check directly that $\frac{{\widehat\mu}^2}{4}$ is (the only non-zero) eigenvalue. This proves $\rho=\frac{{\widehat\mu}^2}{4}\,$.
\end{proof}

\begin{remark}\label{remark_alpha->0}
It is not difficult to realize that Theorem \ref{main_theorem} holds also when $\alpha_r\to0$ for some $r$. Indeed, by \eqref{bound_remainder} and Lemma \ref{concentration_lemma} we see that it is sufficient to require:
\begin{align}
    \alpha_r^2\,\mathbb{E}\Big\langle\left(\mathcal{L}_r-\mathbb{E}\langle\mathcal{L}_r\rangle_{N,t}^{(\epsilon)}\right)^2\Big\rangle_{N,t}^{(\epsilon)}\longrightarrow0\quad\text{as }N\to\infty\,.
\end{align}
The proof of Lemma \ref{concentration_lemma} consists in showing that (see inequalities (A.11) and (A.24) in \cite{us}):
\begin{align}
    &\alpha_r^2\,\mathbb{E}_{\boldsymbol{\epsilon}}\mathbb{E}\Big\langle\left(\mathcal{L}_r-\langle\mathcal{L}_r\rangle_{N,t}^{(\epsilon)}\right)^2\Big\rangle_{N,t}^{(\epsilon)}=O\left(\frac{\alpha_r}{Ns_N^K}\right)\\
    &\alpha_r^2\,\mathbb{E}_{\boldsymbol{\epsilon}}\mathbb{E}\left[\left(\langle\mathcal{L}_r\rangle_{N,t}^{(\epsilon)}-\mathbb{E}\langle\mathcal{L}_r\rangle_{N,t}^{(\epsilon)}\right)^2\right]=O\left(\frac{1}{s_N^{4K/3}N^{1/3}}\right)\,.
\end{align}The previous equalities both vanish in the thermodynamic limit with the choice $s_N\sim N^{-1/16K}$ for instance, independently on $\alpha_r$. Hence the remainder of the interpolation in the proof still goes to $0$ with no variation in the hypothesis.

When a form factor, say $\alpha_r$, vanishes the related component of the order parameter $x_r$ disappears from the variational pressure \eqref{p_var_DBM}. Moreover, if the corresponding $L_r$ is an intermediate layer one can see that the system decouples into two independent DBMs because the effective interaction matrix $\Delta$ becomes block diagonal and the convex $\psi$-term related to the mentioned layer is weighed by $\alpha_r$. The global variational pressure is thus constant in $x_r$ in the thermodynamic limit. 
\end{remark}

\subsection{Proof of Theorem \ref{teor_uniq}}
Uniqueness of the solution of the consistency equation for positive external fields can be proven adapting the strategy used in \cite{noi_deep2}, where the replica symmetric equation of a deep Boltzmann machines was proved to admit a unique solution when the couplings and the external fields are centred Gaussian random variables.
In particular the layers structure permits to \virg{decouple} the interactions as shown in the following

\begin{remark}\label{decoupling}
The consistency equation \eqref{eq:CE} is equivalent to the following:
\be \label{eq:CE-a}
\begin{cases}
\,x_r \,=\, \E\tanh\left(z\,\sqrt{\Theta_r(a)\,x_r+h_r\,} \,+\, \Theta_r(a)\,x_r+h_r \right) & r=1,\dots,K\\[8pt]
\,\alpha_r\,x_r\;a_r \,=\, \alpha_{r+1}\,x_{r+1} & r=1,\dots,K-1
\end{cases}
\ee
where we have introduced the auxiliary variables $a_1,\dots,a_{K-1}>0$ and the functions
\be \label{eq:theta}
\Theta_r(a) \,\equiv\, \begin{cases}
\Delta_{12}\,a_1\  & \textrm{for }r=1 \\[5pt]
\dfrac{\Delta_{r,r-1}}{a_{r-1}} \,+\, \Delta_{r,r+1}\,a_r\ & \textrm{for }r=2,\dots,K-1 \\[10pt]
\dfrac{\Delta_{K-1,K}}{a_{K-1}}\ & \textrm{for }r=K \\[5pt]
\end{cases}\;.
\ee
Indeed, using the definition of the matrix $M$, it can be easily verified that $(M \mathbf{x})_r \,=\, \Theta_r(a)\,x_r$ for $r=1,\dots,K\,$, for $a$ satisfying the second relation in \eqref{eq:CE-a}. 
\end{remark}

The proof of Theorem \ref{teor_uniq} relies on the following

\begin{lemma}\label{uni_mono1D}
Let $z$ be a standard Gaussian random variable. For every $t,h>0$ the equation
\begin{equation} \label{ce1D}
x \,=\, \E \tanh\left(z\,\sqrt{t\,x\!+\!h\,}\,+\,t\,x\!+\!h\right)
\end{equation}
has a unique positive solution that we denote by $x=\bar x(t,h)>0\,$.
Moreover $\bar x$ is strictly increasing as a function of both $t>0$ and $h>0$.
\end{lemma}

\begin{proof}[of Theorem \ref{teor_uniq}]
Equation \eqref{ce1D} rewrites as $x=F(\,t\,x+h)\,$, where $F(h)\equiv\E\tanh(z\sqrt{h}+h)\,$.
By Lemma \ref{Fconcave}, $F$ takes values in (0,1), is strictly increasing and concave. It follows that equation \eqref{ce1D} admits a unique solution in $(0,1)$ and in particular we can show that the function $f(x)\equiv \frac{1}{x}\,F(\,t\,x+h)$ is strictly decreasing for $x>0\,$. Indeed by Lemma \ref{Fconcave} we have:
\begin{align}
x^2\,f'(x) \,&=\, t\,x\,F'(t\,x+h) - F(t\,x+h) \,<0\ \textrm{in }x=0 \;,\\
\frac{\d}{\d x}\,(x^2\,f'(x)) \,&=\, t^2\,x\,F''(t\,x+h) \,<0
\end{align}
hence
\begin{equation} \label{f'}
x^2\,f'(x) \,=\, t\,x\,F'(t\,x+h) - F(t\,x+h) \,<0 \quad\forall\,x>0 \;.
\end{equation}
Now denoting by $\bar x(t,h)$ the unique positive solution of equation \eqref{ce1D}, we can prove its monotonicity with respect to both parameters by differentiating the self-consistent equation
\begin{equation}\label{ce1Dbar}
\bar x(t,h) = F\left(\,t\,\bar x(t,h) +h\right) \;,
\end{equation}
which leads to
\begin{align}
\label{dqdt} \left(1-t\,F'(t\,\bar x+h)\right)\,\frac{\d\bar x}{\d t} \,&=\, \bar x\,F'(t\,\bar x+h) \\
\label{dqdh} \left(1-t\,F'(t\,\bar x+h)\right)\,\frac{\d\bar x}{\d h} \,&=\, F'(t\,\bar x+h) \;.
\end{align}
Lemma \ref{Fconcave} ensures that \eqref{dqdt}, \eqref{dqdh} are positive quantities, hence to conclude it suffices to show that $1-t\,F'(t\,\bar x+h)>0$. Indeed, dividing the inequality \eqref{f'} by $x$, evaluating it at $x=\bar x(t,h)$ and using the self-consistent equation \eqref{ce1Dbar}, one finds precisely:
\begin{equation}
0>\, t\,F'(t\,\bar x+h) - \frac{F(t\,\bar x+h)}{\bar x} \,=\,  t\,F'(t\,\bar x+h) - 1 \;.
\end{equation}
\end{proof}

\begin{proof}[Proof of Theorem \ref{teor_uniq}]
By Lemma \ref{uni_mono1D}, the first line of \eqref{eq:CE-a} is equivalent to:
\be
x_r \,=\, \bar x\,\big(\Theta_r(a),h_r\big) \quad\forall\,r=1,\dots,K
\ee
where $\bar x$ is uniquely defined and strictly increasing with respect to both its arguments.
On the other hand the second line of \eqref{eq:CE-a} rewrites as:
\be
\alpha_1\,x_1\;a_1\cdots a_r \,=\, \alpha_{r+1}\,x_{r+1} \quad\forall\,r=1,\dots,K-1 \;.
\ee
It is convenient to set $X_1(a_1)\,\equiv\, \alpha_1\;\bar x\,\big(\Theta_1(a)\,,\,h_1\big) = \alpha_1\;\bar x\big(\Delta_{1,2}\,a_1\,,\,h_1\big)$ and for $r\geq2$
\be 
X_r\bigg(\frac{1}{a_{r-1}}\,,\,a_{r}\bigg) \,\equiv\, \alpha_r\;\bar x\,\big(\Theta_r(a)\,,\,h_r\big) \,=\,
\alpha_r\;\bar x \bigg(\frac{\Delta_{r,r-1}}{a_{r-1}}+\Delta_{r,r+1}\,a_r\,,\,h_r\bigg) \;.
\ee
Therefore equation \eqref{eq:CE-a} is equivalent to the following:
\be \label{eq:CE-a2}
X_1(a_1)\;a_1\cdots a_r \,=\, X_{r+1}\bigg(\frac{1}{a_{r}}\,,\,a_{r+1}\bigg) \quad\forall\,r=1,\dots,K-1 \,.
\ee
We will show by induction on $r\geq1$ that for any given $a_{r+1}\geq0$ there exists a unique $a_r=\bar a_r(a_{r+1})>0$ such that
\be \label{starp}\begin{cases}
\;a_{r-1} \,=\, \bar a_{r-1}(a_r) \\
\;\vdots \\
\;a_1 \,=\, \bar a_1(a_2) \\[2pt]
\;X_1(a_1)\;a_1\,\cdots\,a_{r-1}\,a_r \,=\, X_{r+1}\bigg(\dfrac{1}{a_r}\,,\,a_{r+1}\bigg) 
\end{cases} \ee
and moreover $\bar a_r$ is a strictly increasing function with respect to $a_{r+1}\,$. The uniqueness of solution of \eqref{eq:CE-a2} will follow immediately by stopping the induction at $r=K-1$ and choosing $a_K=0\,$ and the Theorem will be proven thanks to Remark \ref{decoupling}.\\

$\bullet$ Case $r=1$: given $a_2\geq0$, let's consider the equation
\be \label{star1}
X_1(a_1)\,a_1 \,=\, X_2\bigg(\frac{1}{a_1},a_2\bigg) \;.
\ee
By Lemma \ref{uni_mono1D} the left-hand side of \eqref{star1} is a strictly increasing function of $a_1>0$ and takes all the values in the interval $(0,\infty)$, while the right-hand side is a decreasing function of $a_1>0$ and takes non-negative values.
Therefore there exists a unique $a_1=\bar a_1(a_2)>0$ solution of \eqref{star1}.
Now taking derivatives on both sides of \eqref{star1} and using again Lemma \ref{uni_mono1D}, one finds:
\be 
\frac{d \bar a_1}{d a_2} \,=\,
\frac{\partial}{\partial a_2}X_2\Big(\frac{1}{a_1},a_2\Big) \, \Bigg[\frac{\partial}{\partial a_1}\big(X_1(a_1)\,a_1\big) - \frac{\partial}{\partial a_1}X_2\Big(\frac{1}{a_1},a_2\Big) \Bigg]^{-1}_{\big|a_1=\bar a_1 (a_2)} >0 
\ee
hence $\bar a_1$ is a strictly increasing function of $a_2\,$.\\

$\bullet$ For $r>1\,$, $r-1$ $\Rightarrow$ $r$. Fix $a_{r+1}\geq0\,$.
By inductive hypothesis $\bar a_1,\dots,\bar a_{r-1}$ are well-defined and strictly increasing functions.
Defining the composition $A_l\equiv \bar a_l\circ\dots\circ \bar a_{r-1}$ for every $l=1,\dots,r-1$, equation \eqref{starp} rewrites as:
\be \label{starp2}
\big(X_1\circ A_1\big)(a_r)\; A_1(a_r)\,\cdots\,A_{r-1}(a_r)\; a_r \,=\, X_{r+1}\bigg(\frac{1}{a_r},a_{r+1}\bigg) \;.
\ee
By inductive hypothesis and Lemma \ref{uni_mono1D}, the left-hand side of \eqref{starp2} is a strictly increasing function of $a_r>0$ and takes all the values in the interval $(0,\infty)$, while the right hand-side of \eqref{starp2} is a decreasing function of $a_r>0$ and takes non-negative values. Therefore for every $a_{r+1}\geq0$ there exists a unique $a_r=\bar a_r(a_{r+1})>0$ solution of \eqref{starp2}.
Now taking derivatives on both sides of \eqref{starp2} one finds:
\be \begin{split}
&\frac{d \bar a_r}{d a_{r+1}} \,=\;
\frac{\partial}{\partial a_{r+1}}X_{r+1}\Big(\frac{1}{a_r},a_{r+1}\Big)\; \cdot\\
&\cdot\Bigg[\frac{\partial}{\partial a_r}\bigg(\!\big(X_1\circ A_1\big)(a_r)\;A_1(a_r)\,\cdots\,A_{r-1}(a_r)\;a_r\bigg) - \frac{\partial}{\partial a_r}X_{r+1}\Big(\frac{1}{a_r},a_{r+1}\Big) \Bigg]^{-1}_{\big|\,a_r=\bar a_r(a_{r+1})}
\end{split} \ee
which, using again the inductive hypothesis and Lemma \ref{uni_mono1D}, entails that $\bar a_r$ is a strictly increasing function of $a_{r+1}\,$.
\end{proof}
\vspace{10pt}




\section{Conclusions and perspectives}\label{sect_five}
In this work we have solved the $K$-layer deep restricted Boltzmann machine on the Nishimori line which is an instance of a non-convex multi-species model. The solution consists in the computation of the pressure in the thermodynamic limit which is expressed in terms of an ordinary min-max variational principle over $K$ real positive numbers. The properties of the optimizer show the presence of a phase transition related to the interaction strength and to the relative size of each layer defining the geometry of the system. In particular we discovered that the geometry of the system may tune the phase transition.

A possible way to investigate the model for general values of the parameters would be to test the stability of our results when the system is in a neighborhood of the Nishimori line. We plan to perturb the distribution \eqref{gaussian_couplings_NL} and check under which conditions the replica symmetry property breaks down.

After the completion of this work, paper \cite{Reeves} was brought to our attention where the mutual information for a wide class of inference problems is solved by means of a variational principle. While it is possible to obtain our model as an instance of the one considered there, the variational principle presented has no clear correspondence to ours.
We finally mention that a subsequent work \cite{Mourrat2} 
contains a general result that extends the one in the present paper. In particular the authors compute the limiting free energy with a Hamilton-Jacobi approach which proves to be effective also when dealing with lack of convexity in the interactions.
On the other hand, the simplicity of our setting allows us to carry out a thorough study of the variational formula by locating the phase transition and investigating its dependency on the geometry of the system as in Theorem \ref{phase_transition_thm}, Proposition \ref{prop:rhobound} and Theorem \ref{teor_uniq}.\\

\noindent\textbf{Acknowledgments} The authors thank Jean Barbier and Francesco Guerra for interesting discussions. We acknowledge Jean-Cristophe Mourrat for bringing reference \cite{Reeves} to our attention and an anonymous referee for poiting out to us the preprint \cite{Mourrat2}. P.C. acknowledges support from EU project 952026-Humane-AI-Net.  D.A. and E.M. acknowledge support from Progetto Alma Idea 2018, Università di Bologna.

\printbibliography
\end{document}